\documentclass[11pt,letterpaper]{article}

\usepackage{fullpage,graphicx}%
\usepackage{amsmath,amssymb,amsfonts,physics,hyperref}%
\usepackage{amsthm}%
\usepackage{mathrsfs}%
\usepackage{xcolor}%
\usepackage{listings}%
\usepackage[numbers,sort&compress]{natbib}
\usepackage{xfrac}
\usepackage{caption}
\hypersetup{colorlinks = true, linkbordercolor = white, linkcolor = blue, citecolor = blue}

\theoremstyle{thmstyleone}%
\newtheorem{theorem}{Theorem}
\theoremstyle{thmstyletwo}%
\newtheorem{remark}{Remark}%

\theoremstyle{thmstylethree}%
\newtheorem{definition}{Definition}%

\newcommand{\equa}[1]{Eq.~(\ref{#1})} 

\newcommand{\equasa}[2]{Eqs.~(\ref{#1}) and (\ref{#2})}
   
\newcommand{\eqn}[2]{\begin{gather}
#1
\label{#2}
\end{gather}
}

\begin{document}

\title{\bf Temporal to Spatial Instability in a Flow System:\\ A Comparison}


%
\author{{\bf Antonio Barletta}\\[6pt]
{\small Department of Industrial Engineering,}\\
{\small Alma Mater Studiorum Universit\`a di Bologna,}\\
{\small Viale Risorgimento 2, Bologna, 40136, Italy}\\
\texttt{\small antonio.barletta@unibo.it}}

\date{}

\maketitle

%


\abstract{\noindent
The definitions of temporal instability and of spatial instability in a flow system are comparatively surveyed. The simple model of one-dimensional Burgers' flow is taken as the scenario where such different conceptions of instability are described. The temporal analysis of instability stems from Lyapunov's theory, while the spatial analysis of instability interchanges time and space in defining the evolution variable. Thus, the growth rate parameter for temporally unstable perturbations of a basic flow state is to be replaced by a spatial growth rate when a coordinate assumes the role of evolution variable. Finally, the idea of spatial instability is applied to a Rayleigh-B\'enard system given by a fluid-saturated horizontal porous layer with an anisotropic permeability and impermeable boundaries kept at different uniform temperatures.  \\[10pt]
}



\section{Introduction}\label{sec1}
The instability for a stationary solution of the local balance equations which govern the fluid flow is a cornerstone topic of the research in fluid mechanics over the last century. Its intrinsic importance is due to the close connection with the analysis of transitional flow and, ultimately, with turbulence. There is also a fundamental area of research where fluid mechanics is interrelated to heat transfer and convection. Within this area, the instability of a flow system pinpoints the conditions for the onset of convective cellular patterns as those identified in B\'enard experiments \cite{benard1900etude} and modelled by \citet{rayleigh1916lix}. The typical setup, well-known as the Rayleigh-B\'enard system, is a horizontal layer of fluid or fluid-saturated porous medium where the horizontal boundaries are kept isothermal with heating from below \cite{Rees2000, drazin2004hydrodynamic, straughan2004energy, Straughan}. 

The classical strategy for the flow stability analysis is an implementation of Lyapunov's idea of instability for a mechanical system, where perturbation means defining a slight alteration of a given flow at the initial instant of time and, then, monitoring its evolution in time as caused by the dynamics of the system. The stability or instability relies entirely on the time evolution of an initially imposed perturbation. In fluid mechanics, such an approach is termed temporal analysis of the flow instability as opposed to the spatial analysis. The concept behind the spatial instability analysis is monitoring the evolution in space, typically in the streamwise direction, of a persistent perturbation signal localised at a given spatial position along the flow direction. Stated in these terms, the spatial instability interchanges the roles of space and time relatively to the temporal instability. The practical interest of the spatial instability analysis stems from the aerodynamics of jets in the pioneering papers by \citet{betchov1966spatial} and by \citet{keller1973spatial}. In this field, the more recent papers by \citet{alves2007local} and by \citet{afzaal2015temporal} made significant contributions. An outlook into the research carried out in the area of spatially-developing instability can be found in textbooks such as \citet{schmid2012stability}.

The aim of this paper is to survey the comparison between the temporal instability and the spatial instability by providing a simple example where such concepts are exploited: the one-dimensional Burgers' equation with a driving linear force. The framework of spatial instability is applied to a real-world system made of a horizontal fluid-saturated porous layer with an anisotropic permeability bounded by impermeable horizontal walls kept at different uniform temperatures. This Rayleigh-B\'enard system is studied along the steps illustrated in previous papers \cite{barletta2021spatially, barletta2021time} where the special case of an isotropic saturated porous layer has been analysed.

\section{From the Temporal Analysis to the Spatial Analysis}
The classical approach to the study of the linear instability in fluid mechanics, {\em viz.} the instability to small-amplitude perturbations of a given flow, consists in a direct application of Lyapunov's idea: testing the instability of an equilibrium state, {\em i.e.}, of a stationary solution of the governing equations of fluid flow means slightly altering the initial condition at time $t=0$ and monitoring how this change modifies the system evolution at $t>0$. In this framework, instability means a gradually amplifying discrepancy from the original stationary solution as time evolves. Furthermore, in this well-established approach, a perturbation is a disturbance of the initial condition at $t=0$ with the boundary conditions for the system left unchanged.

\begin{figure}[t]%
\centering
\includegraphics[width=0.7\textwidth]{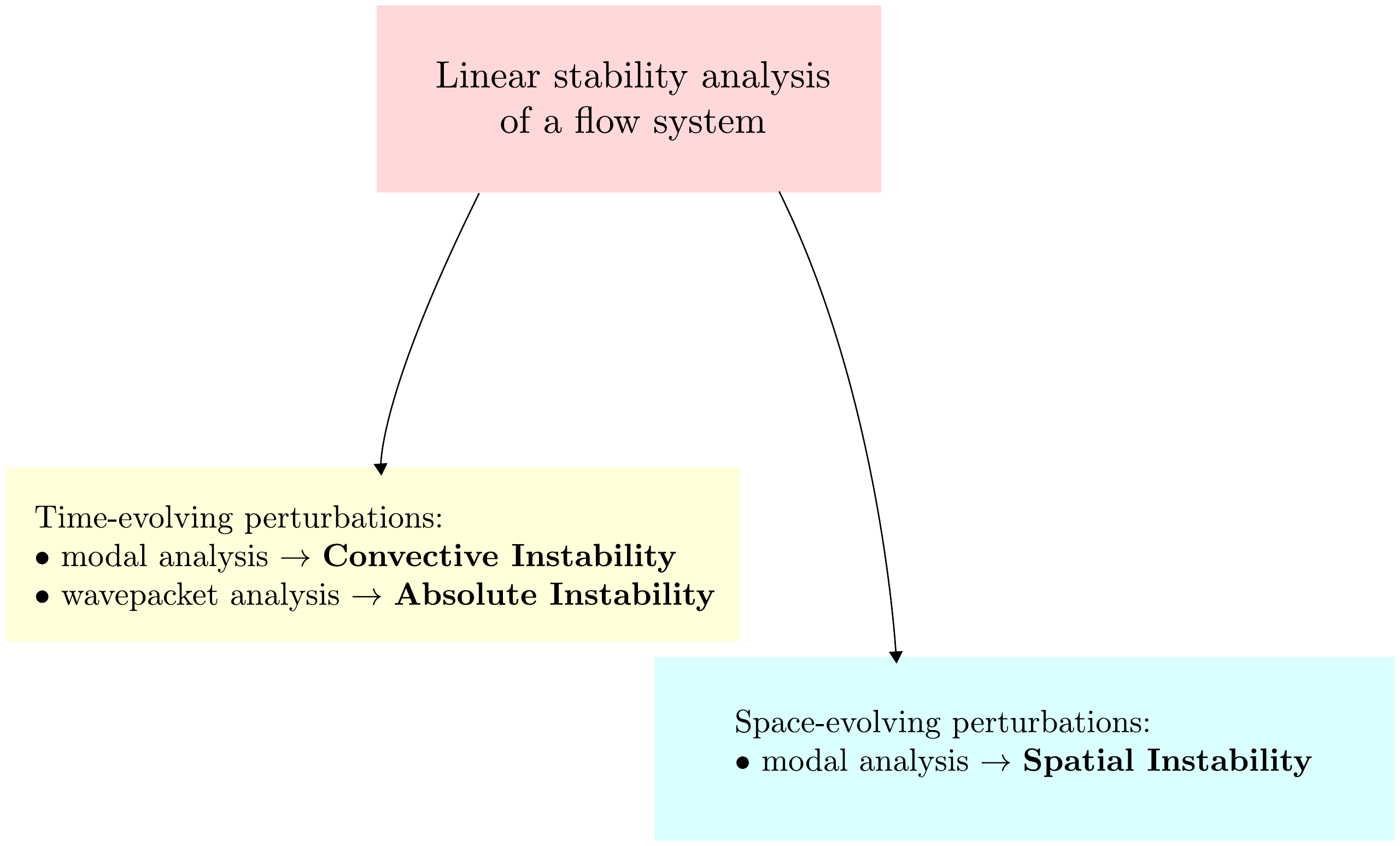}
\caption{The temporal analysis of instability versus the spatial analysis as a comparison between the effects of a time-evolving perturbation and the effects of a space-evolving perturbation}\label{fig1}
\end{figure}

\begin{figure}[h]%
\centering
\includegraphics[width=0.53\textwidth]{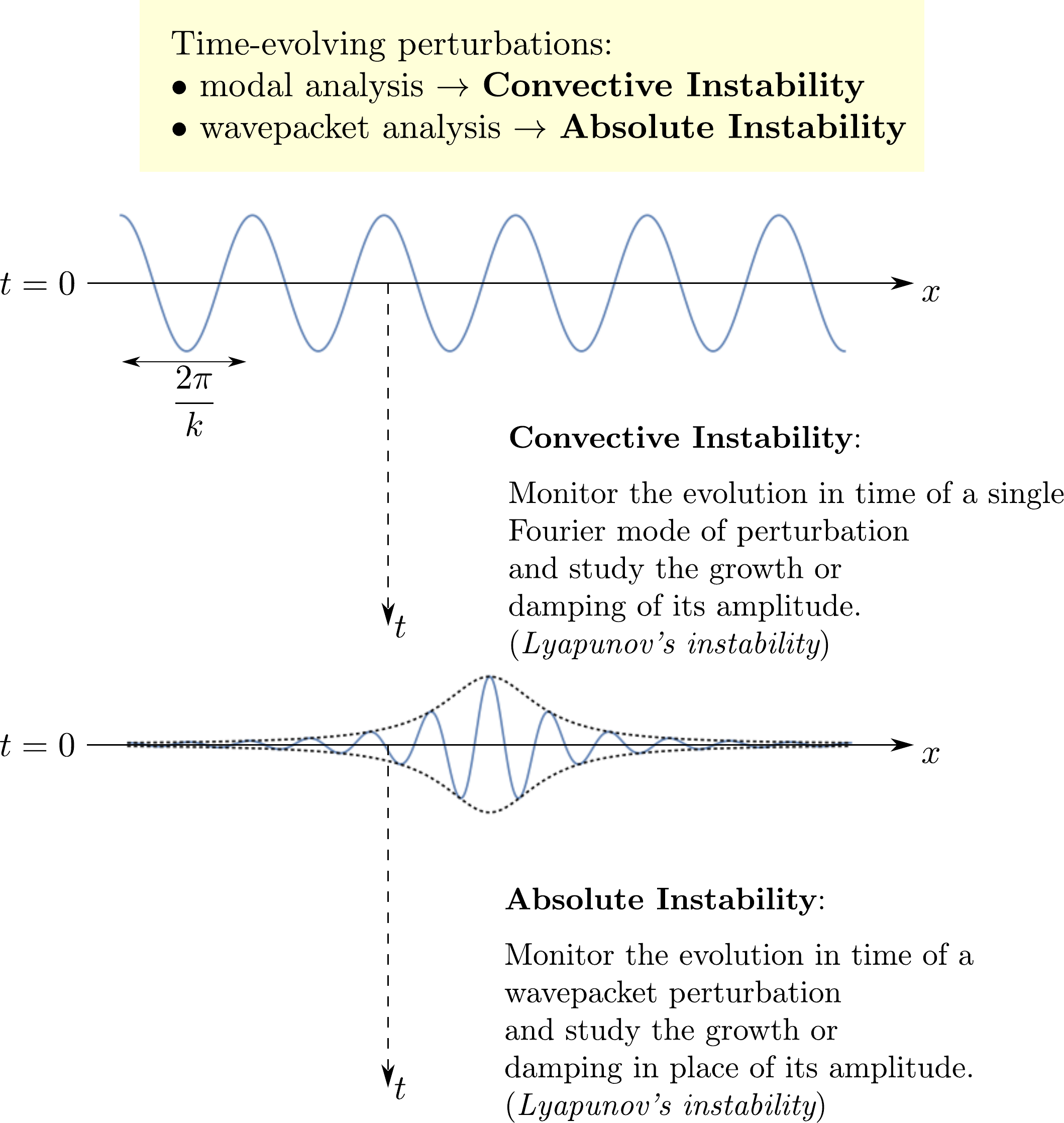}
\caption{The temporal analysis of instability through its two-fold implementation: the convective instability and the absolute instability}\label{fig2}
\end{figure}

Another approach to instability in fluid mechanics, the spatial analysis, marks a sharp difference with respect to the temporal analysis of time-evolving perturbations. In fact, the spatial analysis does not rely on Lyapunov's idea of instability as a way to test the altered time evolution of a system as a consequence of a small change in the initial condition. If we imagine a flow system with a streamwise coordinate $x$, the spatial analysis is meant to test the development along the $x$ direction of a perturbed boundary, or inlet, condition at $x=0$. Then, the aim of the spatial analysis of instability is the determination of the effects produced downstream or upstream of a time-periodic perturbation signal, {\em viz.} a Fourier mode, set at a given position, conventionally at $x=0$. A schematic illustration of the comparison between the analysis of time-evolving perturbations and the analysis of space-evolving perturbations is provided in Fig.~\ref{fig1}.

\begin{figure}[t]%
\centering
\includegraphics[width=0.55\textwidth]{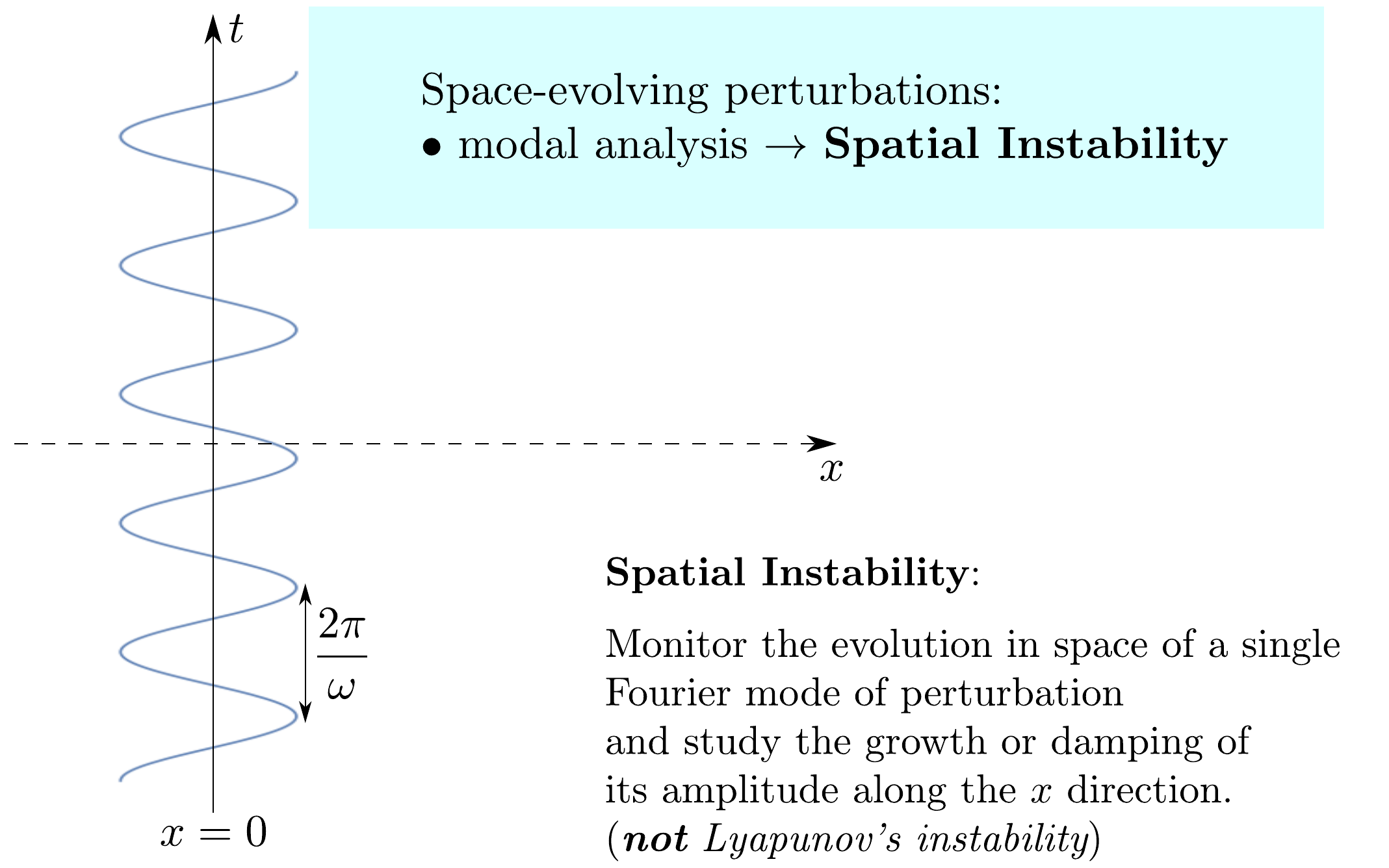}
\caption{The spatial analysis of instability}\label{fig3}
\end{figure}

There are two ways, traditionally employed, for the implementation of the temporal stability analysis and, hence, of the study of time-evolving perturbations: the modal analysis and the wavepacket analysis. The difference relies in the specification of the initial condition set at $t=0$. The modal analysis arises from an initial perturbation expressed as a plane wave with wavenumber $k$ along the $x$ direction, namely a single Fourier mode. On the other hand, the wavepacket analysis is developed by assuming an initial perturbation given by a wavepacket, namely an envelope of infinite plane waves with all possible wavenumbers or, stated differently, a Fourier integral. A scheme of the two approaches to the temporal analysis is shown in Fig.~\ref{fig2}. We mention that the wavepacket analysis leads to the definition of a parametric regime called absolute instability diverse from  the convective instability defined via the modal analysis \cite{barletta2019routes}. 

A sketch illustrating the spatial approach to instability and, hence, the study of space-evolving perturbations is displayed in Fig.~\ref{fig3}. This figure shows the key point of the spatial instability where the response of the flow system is monitored downstream or upstream of a persistent time-periodic perturbation set up at a given spatial position, $x=0$. The time-periodic signal at $x=0$ is, in fact, a Fourier mode having an angular frequency $\omega$.
In Fig.~\ref{fig3}, the already mentioned departure of the spatial instability concept from Lyapunov's idea of instability is also highlighted.

Hereafter, we will focus on the comparison between the temporal analysis and the spatial analysis of instability in their modal formulations. In other terms, the focus will be on convective instability versus spatial instability. On the other hand, no further discussion of the wavepacket dynamics and the absolute instability will be provided here as this topic has been extensively presented in a recent book \cite{barletta2019routes}.

\section{One-Dimensional Burgers' Flow}\label{burflo}
Let us consider a flow system whose dynamics is governed by a one-dimensional Burgers' equation with a driving linear force, namely
\eqn{
\pdv{W}{t} + W \pdv{W}{x} = \pdv[2]{W}{x} + \gamma\ \qty( W - a)  \qc \qty(x,t) \in \mathbb{R}^2,
}{1}
where $\gamma$ and $a$ are two real constants with $a \ne 0$. The assumption $a > 0$ is not restrictive as one can always recover the behaviour for a negative $a$ by applying to \equa{1} the transformation $W \to - W$ with $x \to -x$. Hence, in the following, we will implicitly consider $a > 0$.
Although Burgers' equation is a partial differential equation originally formulated as a toy model for developing turbulence \cite{burgers1948mathematical}, we will not mind about its physical meaning and just use it as an arena for the implementation of the temporal and spatial stability analysis.

Equation~(\ref{1}) admits a simple constant solution,
\eqn{
W(x,t) = a,
}{2}
whose linear instability can be investigated by employing either the temporal analysis or the spatial analysis. In both cases, the starting point is the definition of a small-amplitude perturbation of the equilibrium solution (\ref{2}),
\eqn{
W(x, t) = a + \varepsilon w(x,t),
}{3}
where $\varepsilon$ is a perturbation amplitude parameter. Hence, one can substitute \equa{3} into \equa{1} and neglect terms $\order{\varepsilon^2}$, so that we obtain
\eqn{
\pdv{w}{t} + a \pdv{w}{x} = \pdv[2]{w}{x} + \gamma w .
}{4}

\subsection{Temporal Analysis}\label{temana}
One may employ an $x$-based Fourier transform to solve \equa{4}, namely
\eqn{
\tilde{w}(k,t) = \frac{1}{\sqrt{2 \pi}} \int_{-\infty}^{\infty} w(x,t) e^{-i k x} \dd x 
\qc
w(x,t) = \frac{1}{\sqrt{2 \pi}} \int_{-\infty}^{\infty} \tilde{w}(k,t) e^{i k x} \dd k ,
}{5}
where $k$ has the meaning of a wavenumber. Due to the properties of the Fourier transform for derivatives with respect to $x$, \equasa{4}{5} yield
\eqn{
\pdv{\tilde{w}}{t} = \lambda(k) \tilde{w} , \qfor \lambda(k) = \gamma - k^2 - i k a .
}{6}
Here, the polynomial expression of $\lambda(k)$ is the stability dispersion relation. Furthermore, from \equa{6}, we determine the Fourier transform of the perturbation $w$, namely
\eqn{
\tilde{w}(k,t) = \tilde{w}(k,0) e^{\lambda(k) t} .
}{7}
The analysis of convective instability is modal, which means that the growth or decay in time of the perturbation is assessed for the single Fourier mode which is periodic in the $x$ coordinate with wavelength $2\pi/k$, as sketched in Fig.~\ref{fig2}. The exponential growth or decay of $\tilde{w}(k,t)$ is evidently regulated by the real part of $\lambda(k)$, so that one predicts convective instability, on the basis of \equa{6}, for
\eqn{
\gamma > k^2 .
}{8}
This condition is evidently independent of the constant $a$, while this constant determines the angular frequency $\omega$ of the perturbation wave. In fact, by substituting the expression of $\lambda(k)$ given by \equa{6} into \equa{7}, each single Fourier mode in \equa{5} contains the exponential with imaginary argument
\[
e^{i \qty(k x - \omega t)} \qfor \omega = k a.
\]
Deciding whether the Fourier integral expressing $w(x,t)$ unboundedly grows for large times $t$, at a fixed position $x$, is a different matter which does rely on the steepest-descent approximation of the integral \cite{barletta2019routes}. The answer to this question leads to the condition of absolute instability. We refer the reader to \citet{barletta2019routes} for details on this point. We just mention that satisfying \equa{8} means having an unbounded growth of some Fourier modes, {\em i.e.} those with $|k| < \sqrt{\gamma}$, but this does not imply in general an unbounded growth at large times of the Fourier integral, \equa{5}, expressing $w(x,t)$.  

\subsection{Spatial Analysis}\label{spatana}
One can use a $t$-based Fourier transform to solve \equa{4}, namely
\eqn{
\hat{w}(x, \omega) = \frac{1}{\sqrt{2 \pi}} \int_{-\infty}^{\infty} w(x,t) e^{i \omega t} \dd t  \qc 
w(x,t) = \frac{1}{\sqrt{2 \pi}} \int_{-\infty}^{\infty} \hat{w}(x,\omega) e^{-i \omega t} \dd \omega .
}{9}
On account of the properties of the Fourier transform for the derivative with respect to $t$, \equasa{4}{9} yield
\eqn{
\pdv[2]{\hat{w}}{x} - a \pdv{\hat{w}}{x} + \left( \gamma + i \omega \right) \hat{w} = 0 .
}{10}
The solution of the linear ordinary differential equation at constant coefficients, \equa{10}, is given by
\eqn{
\hat{w}(x,\omega) = c_{+}(\omega) e^{\eta_{+}(\omega)  x} + c_{-}(\omega) e^{\eta_{-}(\omega) x},  
}{11}
where $c_{+}(\omega)$ and $c_{-}(\omega)$ are integration constants determined by the inlet conditions prescribed at $x=0$, while $\eta_{+}(\omega)$ and $\eta_{-}(\omega)$ are the roots of the quadratic equation 
\eqn{
\eta^2 - a \eta + \gamma + i \omega = 0 .
}{12}
%
Let us denote with $s$ and $k$ the real and the imaginary parts of $\eta$, respectively,
\eqn{
\eta = s + i k.
}{14}
Thus, each Fourier mode in the integral representation of $w(x,t)$, \equa{9}, contains the product of two exponential factors
\eqn{
e^{s x} e^{i \qty(k x - \omega t)} .
}{15}
In fact, the product (\ref{15}) involves an amplifying/damping factor, $e^{s x}$. Along the positive $x$ direction, we get amplification when $s > 0$ and damping when $s  < 0$. The other factor in the product (\ref{15}) is the same as that encountered on carrying out the temporal analysis. It defines a travelling wave along the $x$ direction with phase velocity $\omega/k$. Hence, the sign of $\omega/k$ serves to determine whether the wave travels in the positive $(\omega/k > 0)$ or in the negative $(\omega/k < 0)$ direction of the $x$ axis. The difference with respect to the convective instability is apparent. In the temporal framework, we have a one way evolution along the positive direction of $t$, due to the causality principle. In the spatial framework, the evolution can be in the positive $x$ direction or in the negative $x$ direction. Therefore, to assess the spatial instability of a Fourier mode, the sign of the spatial growth rate $s$ is not sufficient, as we also need to know the direction of the $x$ axis where the wave is heading and, hence, the sign of $\omega/k$. One can define the spatial instability by checking if a given Fourier mode grows or decays in the direction of the $x$ axis where this wave is travelling.
\begin{definition}\label{def1}
A Fourier mode with angular frequency $\omega$ and $k \ne 0$ is spatially unstable when
\eqn{
\frac{s \omega}{k} > 0,
}{16}
where $s$ and $k$ are real solutions of
\eqn{
\begin{cases}
s^2 - k^2 - a s + \gamma = 0 \\
2 k s - a k + \omega = 0 
\end{cases} .
}{17}
\end{definition}
It must be noted that \equa{17} is obtained by substituting \equa{14} into \equa{12}. Moreover, for every $\omega$, there are two pairs $(s, k)$ satisfying \equa{17} as there are two complex roots of \equa{12}, {\em i.e.} $\eta_{+}(\omega)$ and $\eta_{-}(\omega)$. We also note that \equasa{16}{17} are left invariant by the transformation
\eqn{
\qty(\omega, s, k) \to \qty(-\omega, s, - k).
}{18}
Hence, the assumption $\omega \ge 0$ is not restrictive as we can always use \equa{18} to extend our considerations to negative angular frequencies.  Thus, we can reformulate the condition for spatial instability expressed by Definition \ref{def1}. In fact, one may focus on the case $\omega \ge 0$ by recognising, from \equa{14}, that $\Im\!\qty(\eta^2) = 2 s k$, where $\Im$ denotes the imaginary part of a complex number.

\begin{figure}[t]%
\centering
\includegraphics[width=0.7\textwidth]{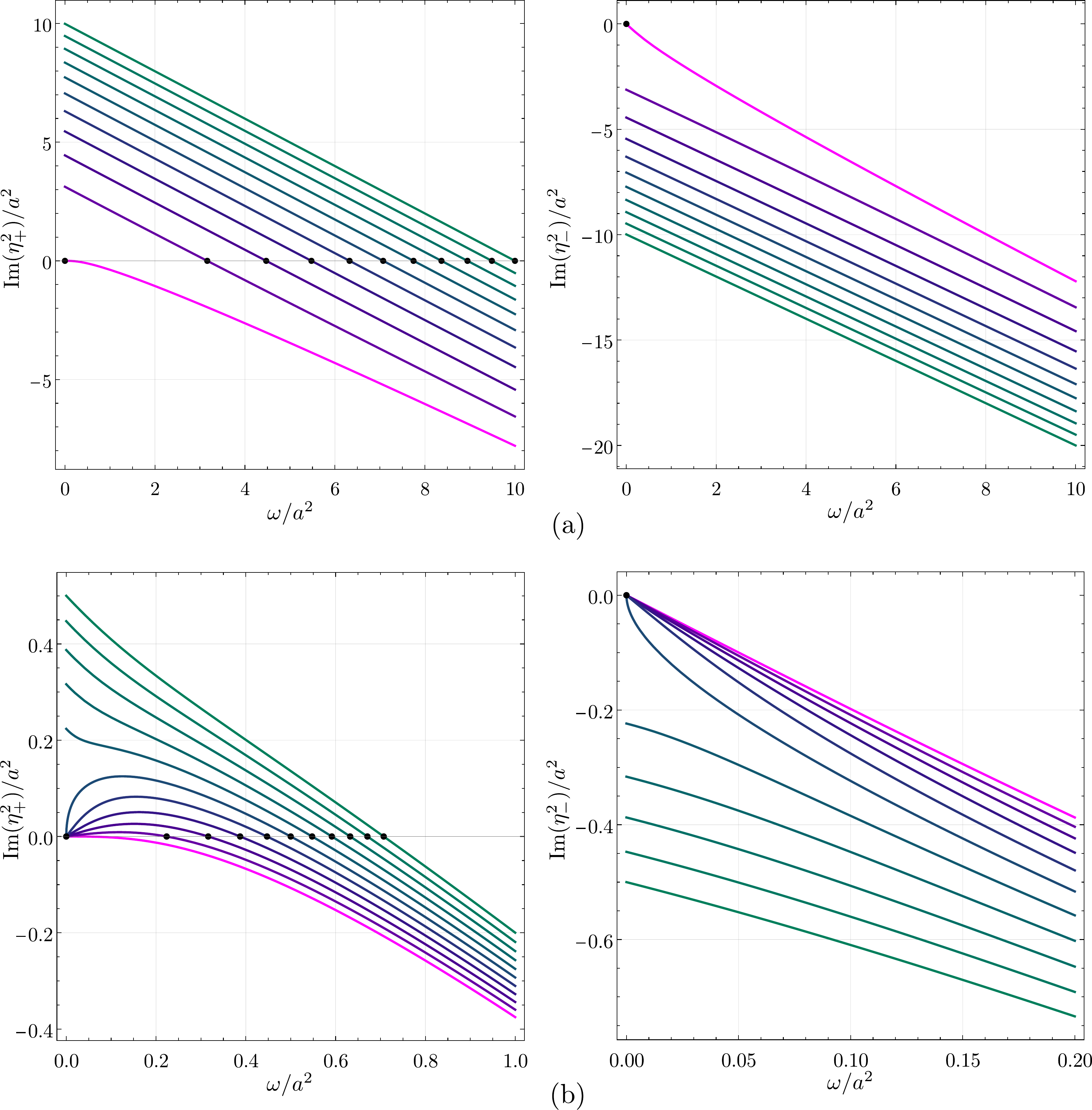}
\caption{Burgers' flow: $\Im\!\qty(\eta^2)/a^2$ versus $\omega/a^2$ for the roots $\eta_{+}$ and $\eta_{-}$ with different values of $\gamma/a^2$: (a) ranging from $0$ to $100$ in steps of $10$; (b) ranging from $0$ to $0.5$ in steps of $0.05$. Both in (a) and in (b), the plot colour continuously changes from cyan $(\gamma/a^2 = 0)$ to green. The black dots denote the zeros for $\omega^2 = \gamma a^2$}\label{fig4}
\end{figure}

\begin{remark}\label{rem1}
A Fourier mode with angular frequency $\omega > 0$ is spatially unstable when
\eqn{
\Im\!\qty(\eta^2) > 0,
}{19}
where $\eta$ is a complex root of \equa{12}.
\end{remark}
Figure~\ref{fig4} shows that only the $\eta_{+}$ root of \equa{12} may satisfy the condition for spatial instability (\ref{19}). Indeed, zeros of $\Im\!\qty(\eta^2) = 2 s k$ are possible for $k=0$ with $s \ne 0$, which yields $\omega=0$ on account of \equa{17}, or for $s=0$ with $k \ne 0$. By employing \equa{17}, one may infer that the latter case leads to  $\omega = a k$ and $\omega^2 = \gamma a^2$. Such values of $\omega$ are denoted with black dots in Fig.~\ref{fig4}.
Frames (b) in Fig.~\ref{fig4} just evidence the small-$\omega$ and small-$\gamma$ parametric region, whereas frames (a) are relative to a significantly larger parametric region.  

An important finding is that $s=0$ defines the subset of the spatial Fourier modes involved in the integral transform (\ref{9}) that is included in the set of temporal Fourier modes discussed in Section~\ref{temana}. More precisely, the spatial modes with $s=0$ coincide with the neutrally stable temporal modes, {\em i.e.} those modes having both a zero spatial growth rate and a zero temporal growth rate.

Equation~(\ref{17}) is the starting point for determining the parametric condition for the onset of the spatial instability.

\begin{theorem}\label{th1}
Spatially unstable modes with $\omega > 0$ can exist only if
\eqn{
\gamma > \frac{\omega^2}{a^2} .
}{20}
\end{theorem}
\begin{proof}
For the sake of brevity, we use the notation
\eqn{
r = \Im\!\qty(\eta^2) = 2 s k . 
}{21}
Then, the second \equa{17} yields
\eqn{
k = \frac{r + \omega}{a} .
}{22}
Substitution into the first \equa{17} leads to
\eqn{
\gamma = \frac{a^4 r (r+2 \omega )+4 (r+\omega )^4}{4 a^2 (r+\omega )^2}.
}{23}
For every given $a$ and $\omega > 0$, the right hand side of \equa{23} can be considered as a function of $r$. Then, we can define 
\eqn{
Y(r) = \frac{a^4 r (r+2 \omega )+4 (r+\omega )^4}{4 a^2 (r+\omega )^2},
}{24}
whose derivative is given by
\eqn{
Y'(r) = \frac{a^4 \omega ^2+4 (r+\omega )^4}{2 a^2 (r+\omega )^3} .
}{25}
Thus, \equa{25} allows one to conclude that $Y'(r) > 0$ for every $r > - \omega$. As a consequence, $Y(r)$ is a monotonic increasing function of $r$ for $r \ge 0$. Thus, on account of \equasa{23}{24}, we conclude that 
\eqn{
r > 0 \quad \Longrightarrow \quad \gamma = Y(r) > Y(0) = \frac{\omega^2}{a^2} .
}{26}
\end{proof}
Theorem~\ref{th1} reveals that, according to the linear analysis, the spatial instability region in the $(\omega, \gamma)$ plane, for $\omega > 0$, is equivalent to the temporal instability region in the $(k, \gamma)$ plane. In fact, the instability region lies in every case above the neutral stability curve which is equivalently given by $\gamma = \omega^2/a^2$ or by $\gamma = k^2$, as a consequence of the equality $\omega = a k$ valid for $s=0$. However, one must bear in mind that the temporal instability and the spatial instability manifest themselves in different manners as, in the former case, the perturbation modes grow exponentially in time and, in the latter case, the perturbation modes grow exponentially in space along their direction of propagation.

Frames (b) of Fig.~\ref{fig4} show that, when $\omega \to 0$, $\Im(\eta^2) \ne 0$ only if $\gamma > a^2/4$. In particular, one has $\Im(\eta^2) = \pm a (\gamma - a^2/4)^{1/2}$ with a spatial growth rate $s = a/2$.
A comment on the possibility to satisfy \equa{17} with $k = 0$, which yields $\Im(\eta^2) = 0$, and $\omega = 0$, is definitely important. Indeed, this case identifies a type of time-independent modes where the spatial growth rate $s$ is a root of 
\eqn{
s^2 - a s + \gamma = 0 .
}{19b}
Real roots of \equa{19b} exist only for $\gamma \le a^2/4$,
\eqn{
s = \frac{a \pm \sqrt{a^2 - 4 \gamma}}{2}.
}{19c}
With a positive $\gamma \le a^2/4$, both roots of \equa{19b} are positive. They yield perturbation modes undergoing a purely exponential growth along the positive $x$ direction. With a negative $\gamma$, \equa{19c} yields a positive and a negative $s$ meaning a growing perturbation mode along the positive $x$ direction and a growing perturbation mode along the negative $x$ direction.
The exponential growth in $|x|$ of such time-independent modes defines a growing departure from the basic equilibrium state as $|x|$ increases. In this sense, the spatial modes with $k=0$ and $\omega=0$ can be classified as spatially unstable even if they do not satisfy \equa{19}, whereas the left hand side of \equa{16} is actually undefined. If one accepts this conception of spatial instability, then spatially unstable modes may exist in a parametric domain $(\gamma < 0)$ where temporal instability is not possible according to a linear analysis.

\section{Horizontal Anisotropic Porous Layer}\label{anpola}
Let us consider a horizontal porous layer with thickness $L$ and infinite horizontal width saturated by a Newtonian fluid. We assume a two-dimensional flow field in the $(x,y)$ plane (see Fig.~\ref{fig5}) with seepage velocity $\vb{u} = \qty(u_x, u_y)$ and temperature $T$. Heating is supplied from below through impermeable and isothermal boundaries kept at different uniform temperatures, $T_1$ and $T_2$. 

The porous material has a uniform, but anisotropic, permeability with the permeability tensor having principal axes along the $x$ and $y$ directions,
\eqn{
\vb{K} = \mqty(K_x & 0 \\ 0 & K_y) .
}{27}

\begin{figure}[t]%
\centering
\includegraphics[width=0.7\textwidth]{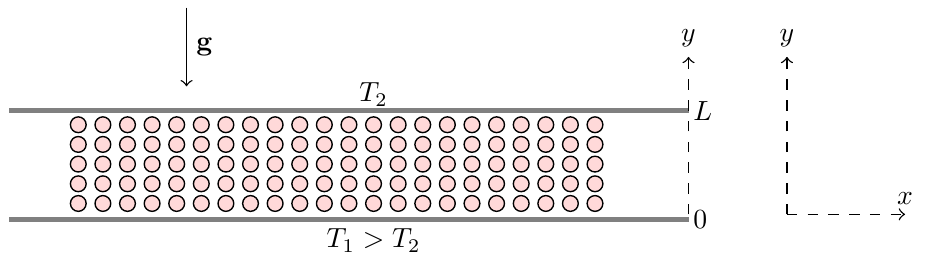}
\caption{Fluid saturated porous layer: sketch of the coordinate system and boundary conditions}\label{fig5}
\end{figure}

\subsection{The Governing Equations}
The momentum transfer is modelled according to Darcy's law \cite{Rees2000, Straughan, NieldBejan2017} and the Boussinesq approximation is claimed in order to model the thermal buoyancy force induced by the non-uniform temperature field. In a dimensionless formulation, the system of local mass, momentum and energy balance equations is given by \cite{Rees2000, Straughan, NieldBejan2017},
\eqn{
\pdv{u_{x}}{x} + \pdv{u_{y}}{y} = 0 ,
\nonumber\\
\pdv{u_{x}}{y} - \tau \pdv{u_{y}}{x} + R \pdv{T}{x} = 0 ,
\nonumber\\
\pdv{T}{t} + u_x \pdv{T}{x} + u_y \pdv{T}{y} = \pdv[2]{T}{x} + \pdv[2]{T}{y} , 
}{28}
where the local momentum balance equation is written in its vorticity formulation, so that the dependence on the pressure field is encompassed. Furthermore, the permeability ratio $\tau$ and the Rayleigh number $R$ are defined as
\eqn{
\tau = \frac{K_x}{K_y} \qc R = \frac{g \beta \qty(T_1 - T_2) K_x L}{\nu \alpha} .
}{29}
Here, $\alpha$, $\beta$ and $\nu$ are the average thermal diffusivity of the saturated porous medium, the thermal expansion coefficient of the fluid and the kinematic viscosity of the fluid, respectively. The gravitational acceleration $\vb{g}$ has a modulus $g$. Due to their definitions, both parameters $\tau$ and $R$ are to be considered as positive.

In order to obtain \equa{28}, we denote with $\sigma$ the ratio between the heat capacity of the saturated porous medium and that of the fluid, while the governing balance equations are made dimensionless by employing the constants
\eqn{
L \qc \frac{L^2 \sigma}{\alpha} \qc \frac{\alpha}{L} \qc T_1 - T_2 ,
}{30}
to scale the coordinates, time, velocity and temperature, respectively. More precisely, the dimensionless temperature is defined as the ratio between $T - T_2$ and the constant $T_1 - T_2$. 

The dimensionless governing equations (\ref{28}) can be reformulated by employing the streamfunction $\Psi$ defined as
\eqn{
u_x = \pdv{\Psi}{y} \qc u_y = - \pdv{\Psi}{x} .
}{31}
In fact, we have
\eqn{
\tau \pdv[2]{\Psi}{x} + \pdv[2]{\Psi}{y} + R \pdv{T}{x} = 0 ,
\nonumber\\
\pdv{T}{t} + \pdv{\Psi}{y} \pdv{T}{x} - \pdv{\Psi}{x} \pdv{T}{y} = \pdv[2]{T}{x} + \pdv[2]{T}{y} .
}{32}
The boundary planes $y=0$ and $y=1$ are considered impermeable and isothermal, namely
\eqn{
\pdv{\Psi}{x} = 0 \qc T = 1 \qfor y = 0,
\nonumber\\
\pdv{\Psi}{x} = 0 \qc T= 0 \qfor y = 1.
}{33}

\subsection{Basic Solution}
A stationary solution of \equasa{32}{33} describing the basic equilibrium state is given by
\eqn{
\Psi = 0 \qc T = 1 - y .
}{34}
%

\subsection{Linearised Perturbation Dynamics}
We define the small-amplitude streamfunction and temperature perturbations of the stationary solution (\ref{34}) as
\eqn{
\Psi(x,y,t) = \varepsilon \psi(x,y,t) \qc T(x,y,t) = 1 - y + \varepsilon \theta(x,y,t) ,
}{36}
where $\varepsilon$ is the perturbation amplitude parameter. Linearisation is carried out by substituting \equa{36} into \equasa{32}{33} and neglecting terms $\order{\varepsilon^2}$. Thus, we can write
\eqn{
\tau \pdv[2]{\psi}{x} + \pdv[2]{\psi}{y} + R \pdv{\theta}{x} = 0 ,
\nonumber\\
\pdv{\theta}{t} + \pdv{\psi}{x} = \pdv[2]{\theta}{x} + \pdv[2]{\theta}{y} ,
\nonumber\\
\pdv{\psi}{x} = 0 \qc \theta = 0 \qfor y = 0, 1.
}{37}
Fourier series in $y$ can be employed, so that the solution of \equa{37} is expressed as
\eqn{
\psi(x,y,t) = \sum_{n=1}^{\infty} \Phi(n,x,t) \sin\!\qty(n \pi y) \qc \theta(x,y,t) = \sum_{n=1}^{\infty} \Theta(n,x,t) \sin\!\qty(n \pi y) .
}{38}
We can now rewrite \equa{37} as
\eqn{
\tau \pdv[2]{\Phi}{x} - n^2\pi^2 \Phi + R \pdv{\Theta}{x} = 0 ,
\nonumber\\
\pdv{\Theta}{t} + \pdv{\Phi}{x} = \pdv[2]{\Theta}{x} - n^2 \pi^2 \Theta .
}{39}

\subsection{Temporal Analysis}
We use the $x$-based Fourier transform defined by \equa{5}, so that \equa{39} is transformed into
\eqn{
\qty(k^2 \tau + n^2\pi^2) \tilde{\Phi} - i k R \tilde{\Theta} = 0 ,
\nonumber\\
\pdv{\tilde{\Theta}}{t} + i k \tilde{\Phi} = - \qty(k^2 + n^2 \pi^2) \tilde{\Theta} .
}{40}
In particular, \equa{40} yields
\eqn{
\pdv{\tilde{\Theta}}{t} + \qty( k^2 + n^2 \pi^2 - \frac{k^2 R}{k^2 \tau + n^2\pi^2} ) \tilde{\Theta} = 0 ,
}{41}
whose solution is
\eqn{
\tilde{\Theta}(n,k,t) = \tilde{\Theta}(n,k,0) e^{\lambda(n, k) t} \qfor \lambda(n, k) = \frac{k^2 R}{k^2 \tau + n^2\pi^2} - k^2 - n^2 \pi^2 .
}{42}
The convective instability arises when the real part of $\lambda(n, k)$ is positive, namely for
\eqn{
R > \frac{\qty(k^2 \tau + n^2 \pi^2) \qty(k^2 + n^2 \pi^2)}{k^2} .
 }{43}
The absolute minimum of $R$ for the onset of the convective instability is obtained with the $n=1$ modes and with the critical values \cite{castinel1974critere, ReesPostelnicu2001, capone2009anisotropy, straughan2019anisotropic, celli2022effects}
 \eqn{
 k = k_c =  \frac{\pi }{\tau^{\sfrac{1}{4}}} \qc R = R_c = \pi ^2 \qty(1 + \sqrt{\tau}\,)^2 .
 }{44}
 The neutral stability curve in the $(k, R)$ plane is the plot of the function of $k$ defined by the right hand side of \equa{43} with $n=1$. As such, it depends on the permeability ratio $\tau$.

 \subsection{Spatial Analysis}
The alternative to the temporal analysis is the use of the $t$-based Fourier transform (\ref{9}). By this method, \equa{39} is transformed to
\eqn{
\tau \pdv[2]{\hat{\Phi}}{x} - n^2\pi^2 \hat{\Phi} + R \pdv{\hat{\Theta}}{x} = 0 ,
\nonumber\\
- i \omega \hat{\Theta} + \pdv{\hat{\Phi}}{x} = \pdv[2]{\hat{\Theta}}{x} - n^2 \pi^2 \hat{\Theta} .
}{46}
Equation~(\ref{46}) is a system of ordinary differential equations in the independent variable $x$. It's solution can be expressed as
\eqn{
\hat{\Phi}(n,x,\omega) = R \sum_{j=1}^4 \frac{\eta_j(n,\omega) C_j (n, \omega) e^{\eta_j(n,\omega) x}}{n^2 \pi^2 - \tau \qty[\eta_j(n,\omega)]^2}  ,
\nonumber\\
\hat{\Theta}(n,x,\omega) = \sum_{j=1}^4 C_j (n, \omega) e^{\eta_j(n,\omega) x} ,
}{47}
where $\eta_j (n,\omega)$, with $j = 1, \ldots, 4$, coincides with a root $\eta$ of the fourth-degree equation
\eqn{
\qty(n^2 \pi^2 - \tau \eta^2) \qty(n^2 \pi^2 - \eta^2 - i \omega ) + R  \eta^2 = 0 ,
}{48}
while $C_j (n, \omega)$ are coefficients to be determined on the basis of the conditions prescribed at $x=0$ for the perturbations $\psi$ and $\theta$.

The modal stability analysis can be scaled down to the case $n=1$ by recognising that the transformation defined by 
\eqn{
\frac{\eta}{n} \to \eta \qc \frac{\omega}{n^2} \to \omega \qc \frac{R}{n^2} \to R ,
}{49}
maps \equa{48} into its $n=1$ version, namely
\eqn{
\qty(\pi^2 - \tau \eta^2) \qty(\pi^2 - \eta^2 - i \omega ) + R  \eta^2 = 0 .
}{50}
Hence, our forthcoming discussion will rely on \equa{50}.
 
\begin{figure}[t]%
\centering
\includegraphics[width=0.7\textwidth]{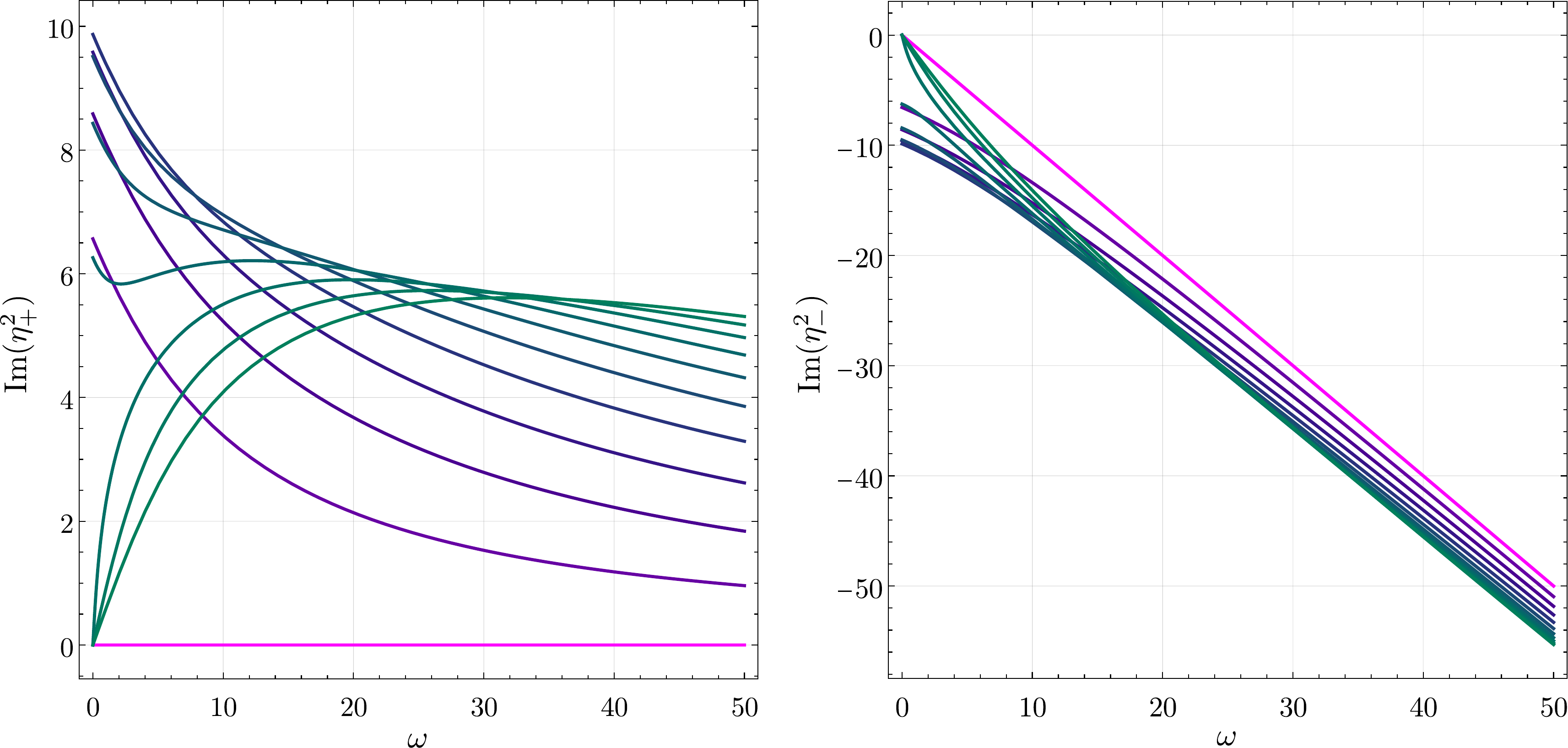}
\caption{Porous medium with $\tau = 1$: $\Im\!\qty(\eta^2)$ versus $\omega$ for the roots $\eta_{+}^2$ and $\eta_{-}^2$ with different values of $R$ ranging from $0$ to $50$ in steps of $5$ with the plot colour continuously changing from cyan $(R=0)$ to green}\label{fig6}
\end{figure}

\begin{figure}[t]%
\centering
\includegraphics[width=0.7\textwidth]{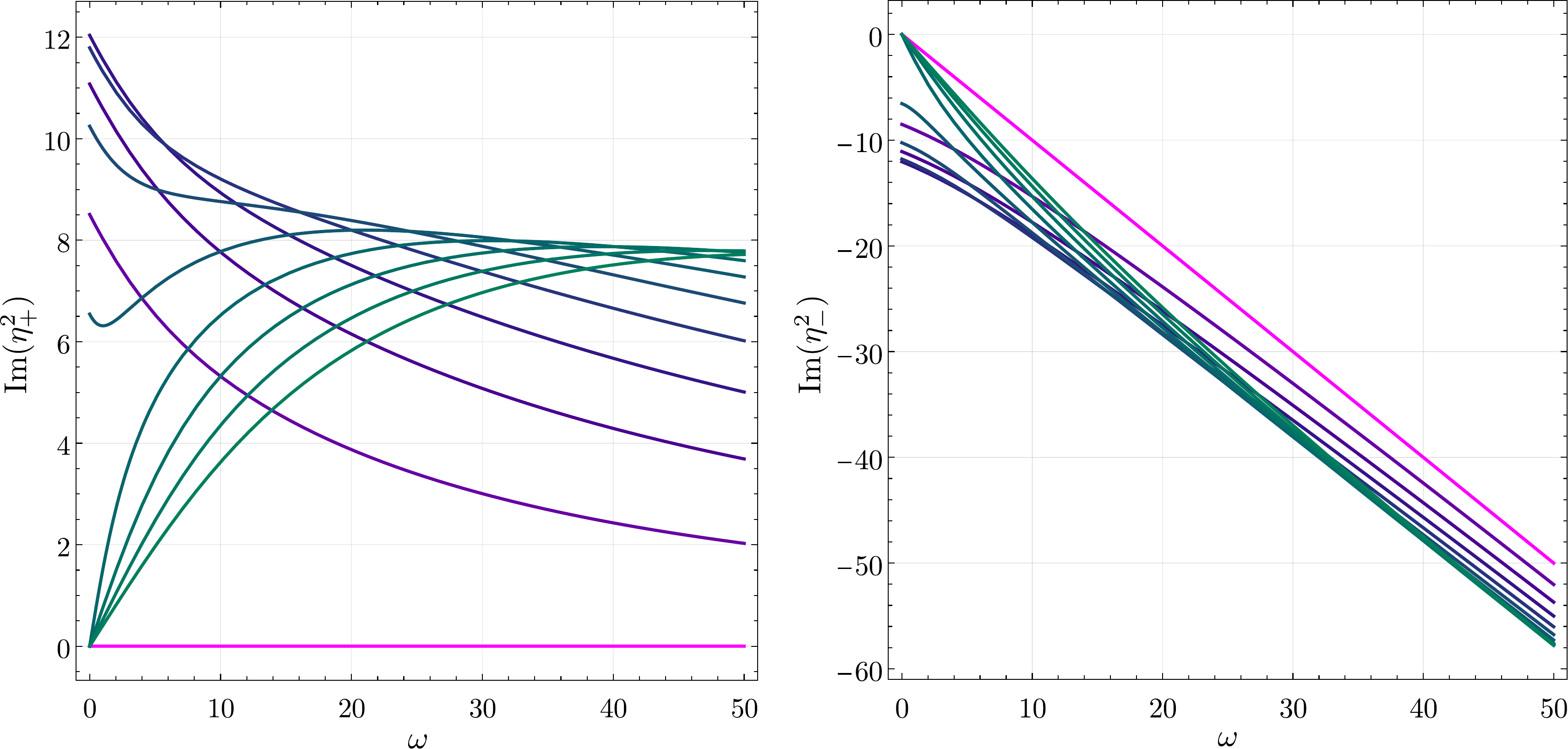}
\caption{Porous medium with $\tau = 2/3$: $\Im\!\qty(\eta^2)$ versus $\omega$ for the roots $\eta_{+}^2$ and $\eta_{-}^2$ with different values of $R$ ranging from $0$ to $50$ in steps of $5$ with the plot colour continuously changing from cyan $(R=0)$ to green}\label{fig7}
\end{figure}

\begin{figure}[t]%
\centering
\includegraphics[width=0.7\textwidth]{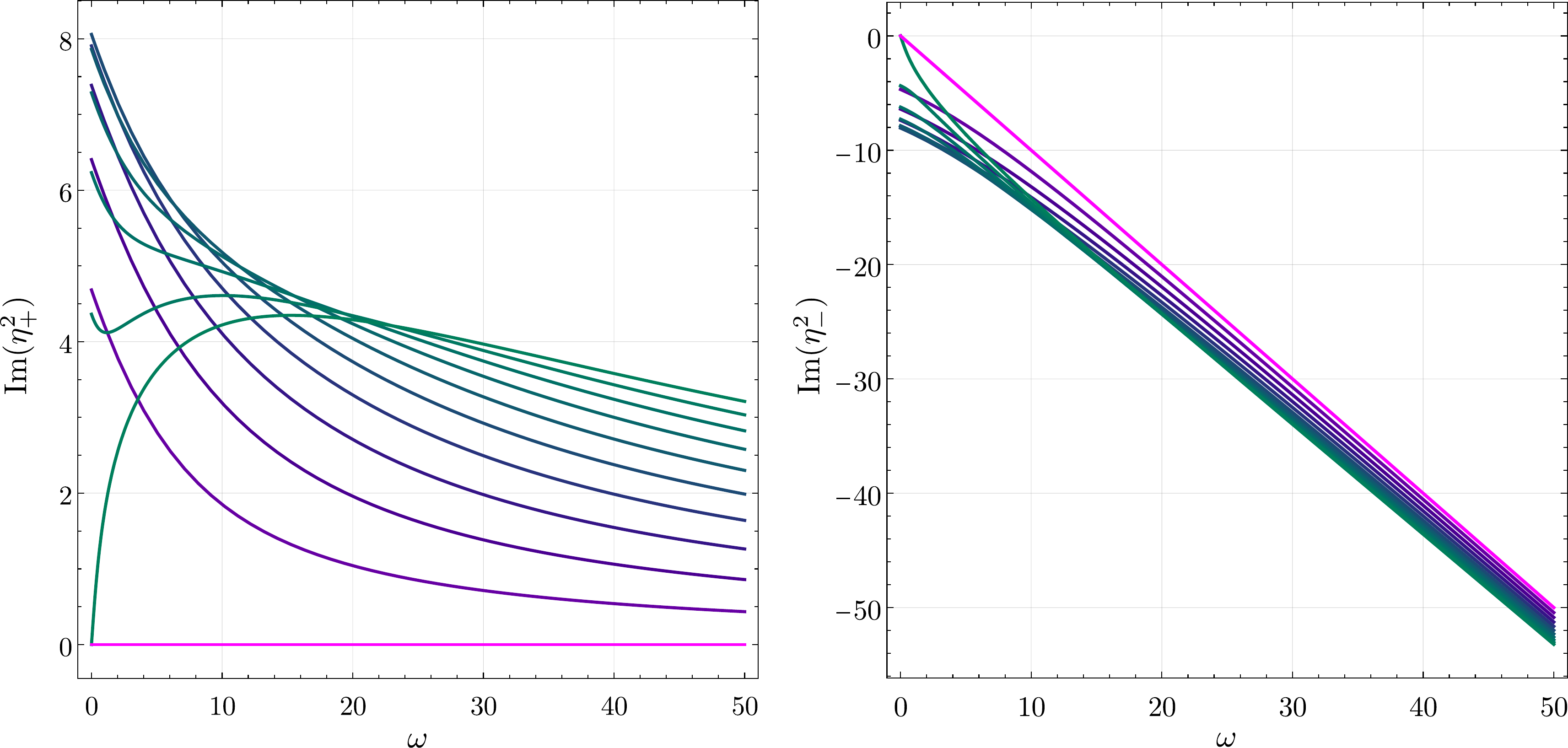}
\caption{Porous medium with $\tau = 3/2$: $\Im\!\qty(\eta^2)$ versus $\omega$ for the roots $\eta_{+}^2$ and $\eta_{-}^2$ with different values of $R$ ranging from $0$ to $50$ in steps of $5$ with the plot colour continuously changing from cyan $(R=0)$ to green}\label{fig8}
\end{figure}

We note that \equa{50} is endowed with the same invariance defined, for Burger's flow, by \equa{18}. This means that we can just focus on positive values of $\omega$, exactly as for Burger's flow.

\begin{figure}[t]%
\centering
\includegraphics[width=0.7\textwidth]{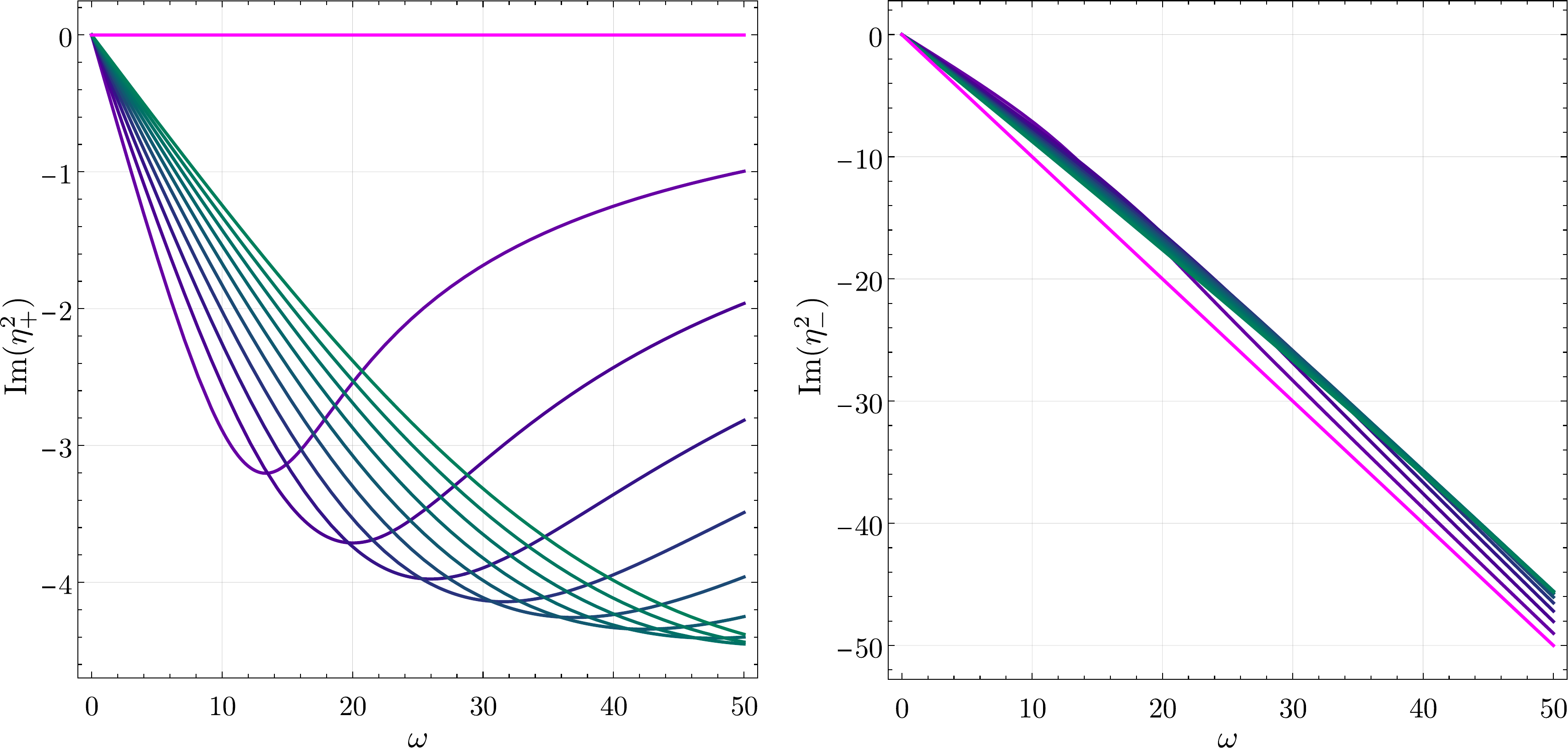}
\caption{Porous medium with $\tau = 1$: $\Im\!\qty(\eta^2)$ versus $\omega$ for the roots $\eta_{+}^2$ and $\eta_{-}^2$ with different values of $R$ ranging from $0$ to $-50$ in steps of $-5$ with the plot colour continuously changing from cyan $(R=0)$ to green}\label{fig9}
\end{figure}

\begin{figure}[h!]%
\centering
\includegraphics[width=0.7\textwidth]{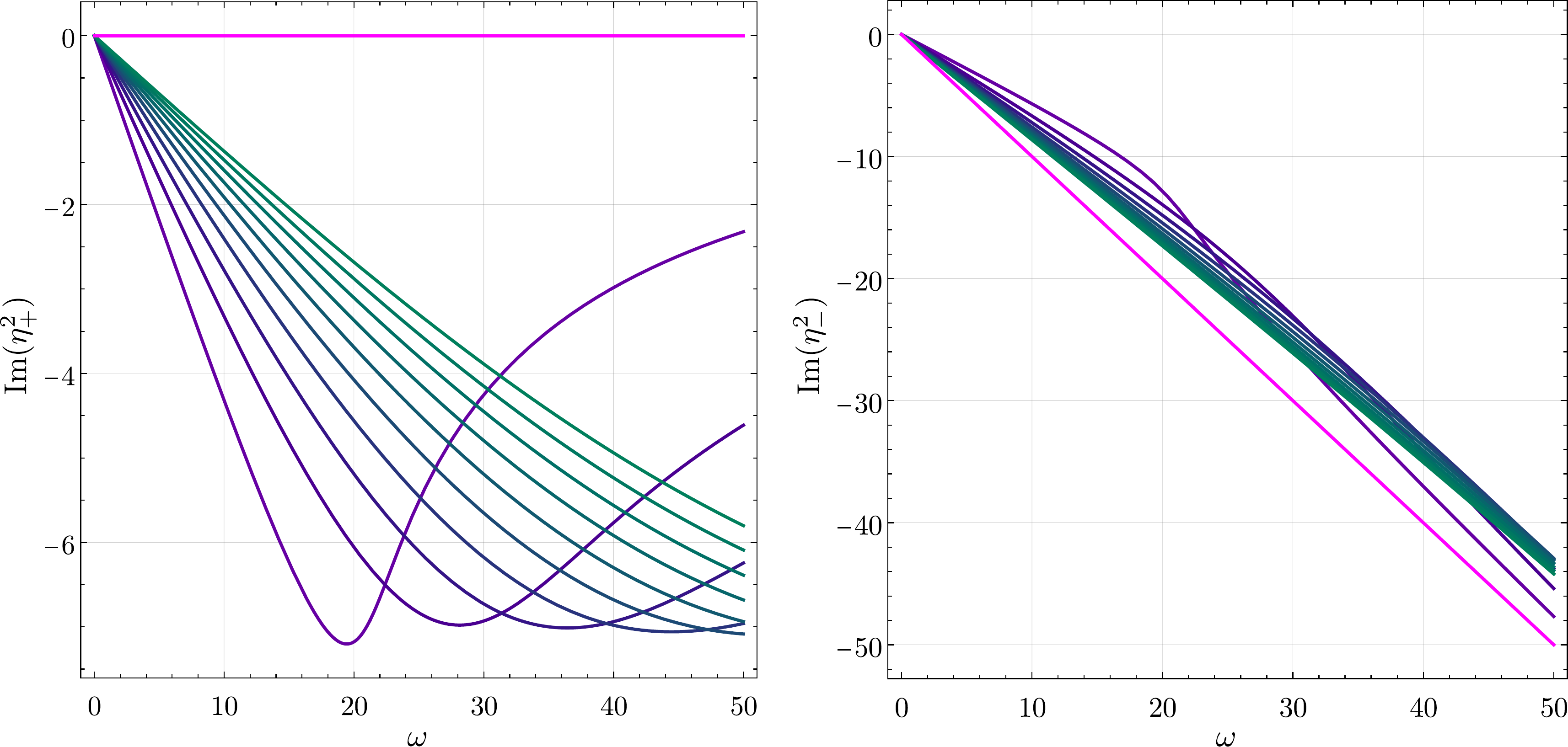}
\caption{Porous medium with $\tau = 2/3$: $\Im\!\qty(\eta^2)$ versus $\omega$ for the roots $\eta_{+}^2$ and $\eta_{-}^2$ with different values of $R$ ranging from $0$ to $-50$ in steps of $-5$ with the plot colour continuously changing from cyan $(R=0)$ to green}\label{fig10}
\end{figure}

\begin{figure}[h!]%
\centering
\includegraphics[width=0.7\textwidth]{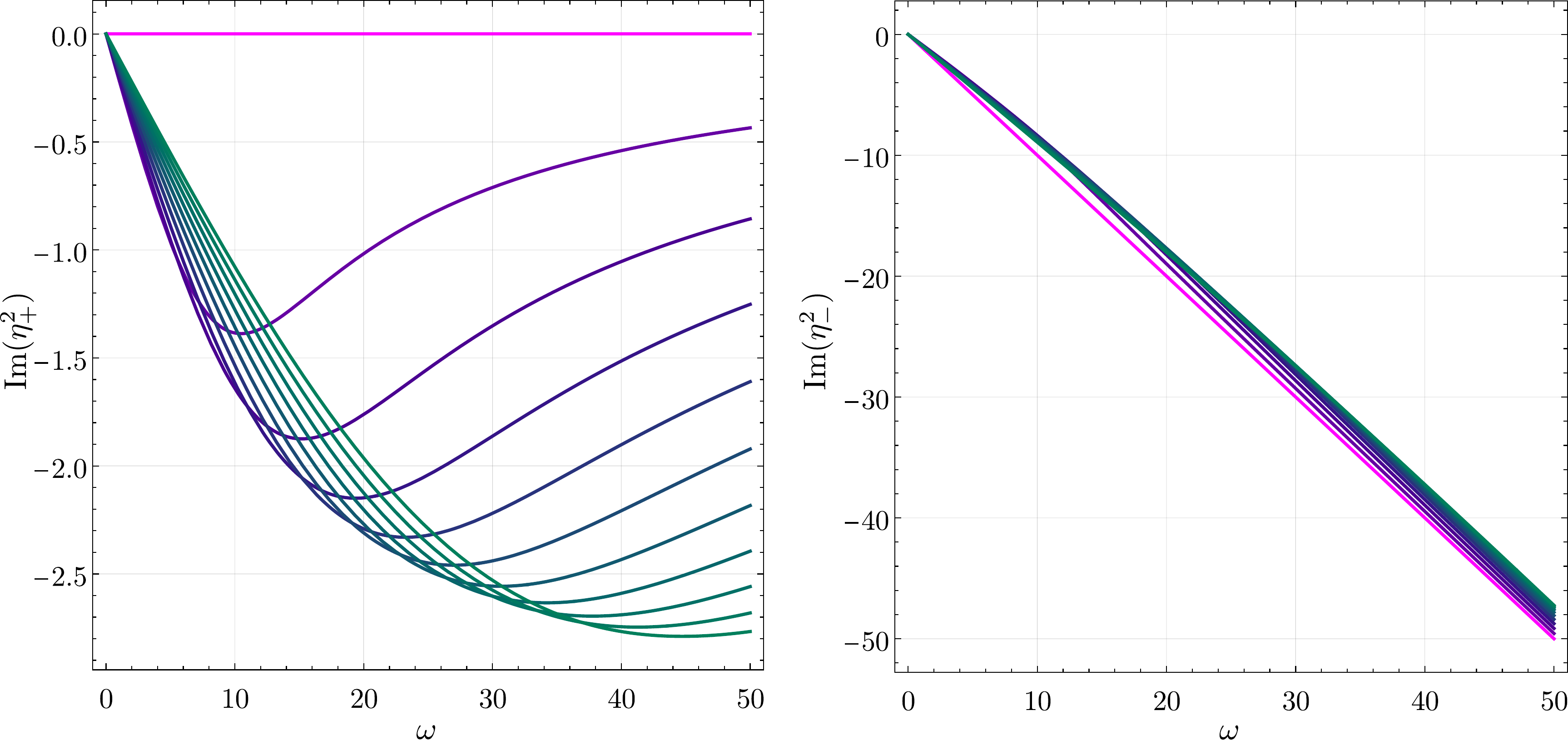}
\caption{Porous medium with $\tau = 3/2$: $\Im\!\qty(\eta^2)$ versus $\omega$ for the roots $\eta_{+}^2$ and $\eta_{-}^2$ with different values of $R$ ranging from $0$ to $-50$ in steps of $-5$ with the plot colour continuously changing from cyan $(R=0)$ to green}\label{fig11}
\end{figure}

Figures~\ref{fig6}-\ref{fig8} display plots of $\Im(\eta^2)$ versus $\omega$ for the two roots of \equa{50} and several values of $R$. These figures are relative to different anisotropy ratios $\tau$, with Fig.~\ref{fig6} describing the isotropy limiting case $(\tau = 1)$. The isotropic medium is a special case already investigated by \citet{barletta2021spatially}. Figure~\ref{fig7} is for $\tau = 2/3$, a case where the porous medium is more permeable in the vertical direction than in the horizontal direction. Figure~\ref{fig8} is for $\tau = 3/2$, {\em i.e.} the porous medium is more permeable in the horizontal direction than in the vertical direction. In Figs.~\ref{fig6}-\ref{fig8}, the cyan curves denote the limiting case where $R \to 0$. As it is easily gathered from \equa{50}, there are two possible roots for $R \to 0$: one is $\eta^2 = \pi^2/\tau$ leading to $\Im(\eta^2) = 0$, the other root yields $\Im(\eta^2) = - \omega$. The former root is not acceptable as it would lead to a vanishing denominator in \equa{47}. Then, the corresponding plot is reported just as a reference case in Figs.~\ref{fig6}-\ref{fig8}. 

The most important information disclosed by Figs.~\ref{fig6}-\ref{fig8} is that, in every case, spatially unstable modes exist for every $R > 0$. In fact, as pointed out in Remark~\ref{rem1}, spatially unstable modes are identified by the condition $\Im(\eta^2) > 0$. Having in mind \equasa{43}{44}, which provide a statement of the condition for temporal instability and reveal that such a condition is allowed only when $R > R_c = \pi ^2 (1 + \sqrt{\tau}\,)^2$, Figs.~\ref{fig6}-\ref{fig8} show that one may have $\Im(\eta^2) > 0$ also for $0 < R < R_c$. We have reached the conclusion that the spatial instability is not equivalent to the temporal instability.

The analysis carried out so far revealed that every setup having $R > 0$, may display modes satisfying the condition $\Im(\eta^2) > 0$ with $\omega > 0$ and, hence, the condition of spatial instability is satisfied for every case with heating from below, no matter how small is $R$. We also observed that $R =0$ yields either $\Im(\eta^2) = 0$ or $\Im(\eta^2) = - \omega$ and, thus, it does not show up any spatial instability. Heating from above, {\em i.e.} $R < 0$, is a parametric condition explicitly excluded in our initial description as we claimed $T_1 > T_2$ (see Fig.~\ref{fig5}), but considering $R < 0$ may offer some further insights anyway. 
%
%
In fact, Figs.~\ref{fig9}-\ref{fig11} display the plots of $\Im(\eta^2)$ versus $\omega$ for $R \le 0$ ranging from $0$ to $-50$. These figures are relative to the isotropic case $\tau = 1$ (Fig.~\ref{fig9}), $\tau = 2/3$ (Fig.~\ref{fig10}) and $\tau = 3/2$ (Fig.~\ref{fig11}). Exactly as for Figs.~\ref{fig6}-\ref{fig8}, all Figs.~\ref{fig9}-\ref{fig11} contain two frames: one for each root, either $\eta_{+}^2$ or $\eta_{-}^2$, of \equa{50}. For all the three cases, $\tau = 1$, $\tau = 2/3$ and $\tau = 3/2$, devised in Figs.~\ref{fig9}-\ref{fig11}, we have always $\Im(\eta^2) \le 0$ suggesting that no spatial instability is possible. 
%
However, as pointed out in Section~\ref{burflo} for the analysis of Burgers' flow, one has to be quite careful in the evaluation of the limiting case of stationary modes, $\omega \to 0$. 

It must be emphasised that, on account of \equa{50}, the limit for $\omega \to 0$ of $\Im(\eta^2)$ can be nonzero only if
\eqn{
\pi ^2 \qty(1 - \sqrt{\tau}\,)^2 < R < \pi ^2 \qty(1 + \sqrt{\tau}\,)^2 ,
}{52}
where the upper bound of the interval is precisely the critical value $R_c$ given by \equa{44}.
Consistently with \equa{52}, Fig.~\ref{fig6} displays $\Im(\eta^2) \to 0$ when $\omega \to 0$ only for $R = 0$ and $R \ge R_c \approx 39.5$, namely for $R = 0$, $40$, $45$ and $50$. Similarly, Fig.~\ref{fig7} displays $\Im(\eta^2) \to 0$ only for $R = 0$ and $R \ge R_c \approx 32.6$, namely for $R = 0$, $35$, $40$, $45$ and $50$, while Fig.~\ref{fig8} displays $\Im(\eta^2) = 0$ only for $R = 0$ and $R \ge R_c \approx 48.8$, that is for $R = 0$ and $50$. On the other hand, for all cases reported in Figs.~\ref{fig9}-\ref{fig11} relative to $R \le 0$, one has $\Im(\eta^2) \to 0$ when $\omega \to 0$. This is expected as a negative $R$ always lies outside the interval defined by \equa{52}. Following the same reasoning presented for the analysis of Burgers' flow, we may consider the cases where the limit $\omega \to 0$ leads to $k=0$, being $s$ and $k$ the real part and the imaginary part of $\eta$, as reported in \equa{14}. This situation undoubtedly means that  the limit $\omega \to 0$ yields $\Im(\eta^2) = 2 k s \to 0$. Hence, spatial modes with both $\omega=0$ and $k=0$ exist only when $R$ is outside the interval defined by \equa{52}. By employing \equa{48}, one can determine $s^2$ for such modes as
\eqn{
s^2 = \frac{- \qty[R - \pi^2 \qty(1 + \tau)] \pm \sqrt{\qty[R - \pi^2 \qty(1 + \tau)]^2 - 4 \pi^4 \tau}}{2 \tau}.
}{53}
The right hand side of \equa{53} is real and non-negative, with both $+$ and $-$, provided that
\eqn{
R \le \pi^2 \qty(1 - \sqrt{\tau})^2 .
}{54}
Therefore, we can have spatial modes with $\omega=0$ and $k=0$ with either a positive or a negative $s$ whenever \equa{54} is satisfied. Such modes are stationary, have $\Im(\eta^2) = 0$, but they display an exponential growth in space either in the positive or in the negative $x$ direction. Exactly as for Burgers' flow, we can consider these modes as a manifestation of spatial instability. The important point is that they do exist also for conditions of heating from above $(R  < 0)$. Figure~\ref{fig12} displays a scheme of the different instability types in the  $(\tau, R)$ parametric plane. The parametric region where growing spatial modes with $\omega=0$ and $k=0$ exist alongside spatially unstable modes with $\omega > 0$ and $\Im(\eta^2) > 0$ is labelled as ``spatial instability with $\omega \ge 0$'' for brevity. The parametric region $R < 0$ where growing spatial modes with $\omega=0$ and $k=0$ exist is labelled as ``spatial instability with $\omega = 0$''. 

\begin{figure}[t]%
\centering
\includegraphics[width=0.5\textwidth]{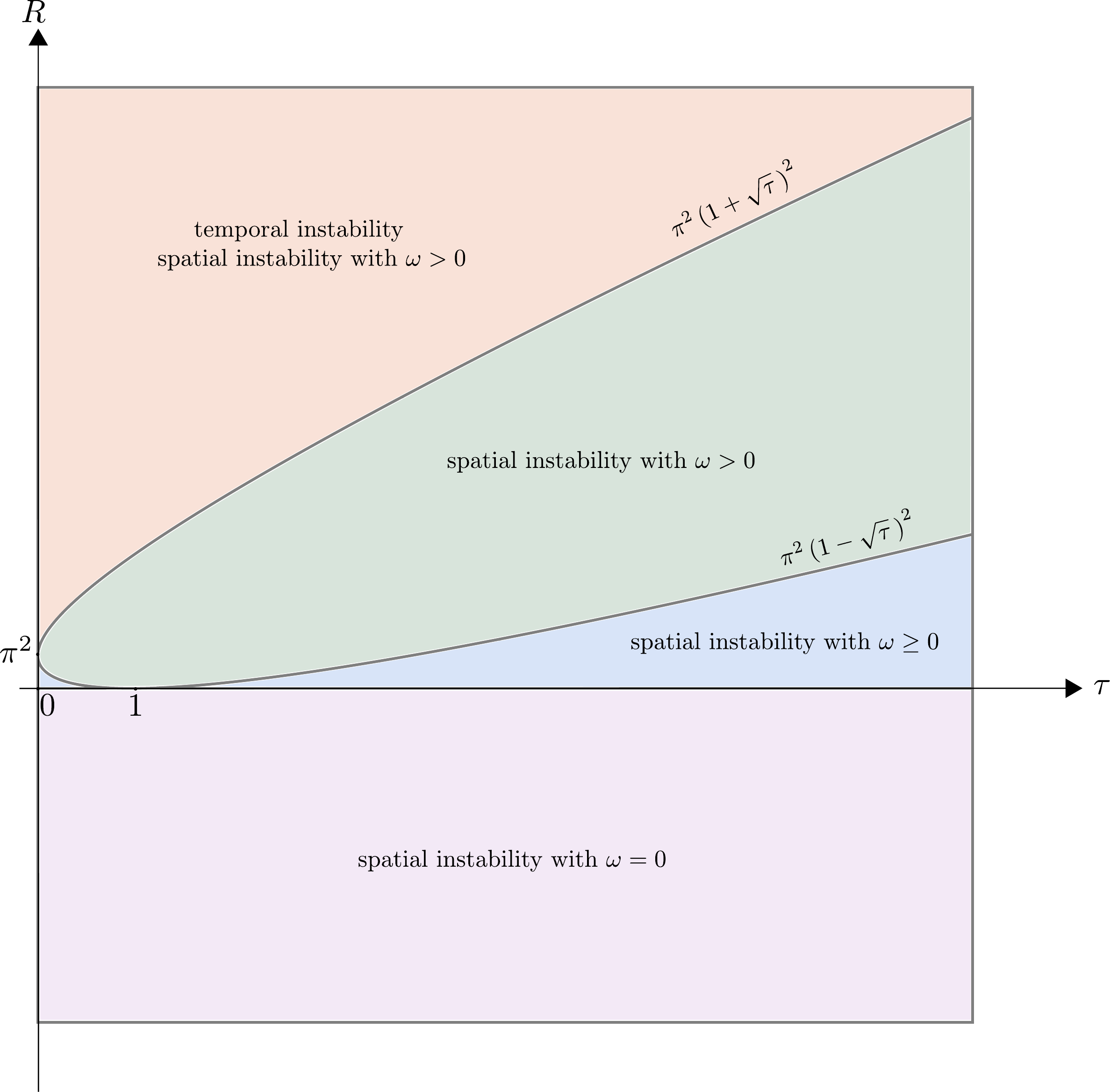}
\caption{Porous medium: instability diagram in the $(\tau, R)$ parametric plane}\label{fig12}
\end{figure}

\section{The Scope of the Spatial Instability Analysis}
The study of two test cases, a toy model based on Burgers' equation and a real-world system modelled as a fluid saturated porous medium with anisotropic permeability, have shown that the linear analysis of the perturbations within either a temporal instability concept or a spatial instability concept may lead to diverse scenarios. The temporal instability, in its modal formulation, allows one to evaluate the convective instability threshold bounding the unstable parametric region through the neutral stability curve and the critical values of the governing parameters. The spatial analysis, in its modal formulation, 
revealed conditions of linear instability for cases where the temporal analysis predicts linear stability. In the examples discussed so far, we found that the spatial instability may always exist either with the Fourier spatial modes having $\omega \ne 0$ or, at the very least, with the stationary spatial modes obtained in the limit $\omega \to 0$. Then, a question arises: does the linear analysis of spatial perturbation modes always lead to the prediction of instability no matter which values are prescribed for the parameters governing the system? The answer is negative.
%
%
%
In fact, a neat example is provided by the one-dimensional diffusion equation,
\eqn{
\pdv{u}{t} = \pdv[2]{u}{x} .
}{55}
An equilibrium state, $u = u_0$, is any stationary solution of \equa{55} and, hence, any linear function of $x$. If we introduce a perturbation of the equilibrium state, $u = u_0 + \varepsilon U$, the linear nature of \equa{55} leads to a governing equation for the perturbation function $U$ which coincides with \equa{55} where $u$ is now replaced by $U$. The convective instability analysis, carried out following the steps and the notations introduced in Section~\ref{temana}, leads us to the temporal dispersion relation,
\eqn{
\lambda(k) = - k^2.
}{56}
The well-known conclusion is utterly evident: every equilibrium state $u_0$ is stable as the real part of $\lambda$, and hence the time growth-rate, is negative for every $k \ne 0$. The spatial instability analysis can be developed according to the steps and notations introduced in Section~\ref{spatana}, leading us to the spatial dispersion relation,
\eqn{
\qty[\eta(\omega)]^2 = - i \omega.
}{57}
We have $\Im(\eta^2) = - \omega < 0$ for $\omega > 0$ so that, by adopting the criterion defined in Remark~\ref{rem1}, again we reach the conclusion that every equilibrium state $u_0$ is stable. Extending the spatial instability analysis to the limit of stationary modes, $\omega \to 0$, does not lead to any further insights. In fact, in this limit, \equa{57} just yields $k=0$ and $s=0$.

The diffusion equation example provides an illustration of the selective nature of the spatial instability analysis. Despite what we concluded for the two test cases developed in Sections~\ref{burflo} and \ref{anpola}, the spatial instability condition is not satisfied in every case.


%

\section{Conclusions and Final Remarks}
The onset of the linear instability in a flow system has been outlined by focussing on two diverse schemes termed temporal instability and spatial instability. If the former scheme is an exploitation of Lyapunov's idea of instability in a mechanical system, the latter scheme is a an utterly different approach where the roles of space and time are interchanged on defining the evolution variable for the development of the instability. Indeed, Lyapunov's or temporal instability can be described as the 
amplification in time of an initial perturbation superposed at $t = 0$ to the equilibrium state. Spatial instability can be described as the 
amplification in the $x$ spatial direction of an inlet perturbation localised at $x = 0$. Within a modal approach, the temporal instability analysis tests the reaction of the system to space-periodic perturbations with a given wavenumber $k$, while spatial instability analysis tests the reaction of the system to time-periodic perturbations with a given angular frequency $\omega$.

The comparison between the temporal instability analysis and the spatial instability analysis has been introduced via a simple flow system governed by a one-dimensional Burgers' equation with a driving linear force. The condition of spatial instability has been formulated through an inequality where the imaginary part of the square complex growth rate $\eta$ must be positive for positive values of $\omega$. Such an inequality reflects the condition of spatial instability as a spatial amplitude growth of the modal perturbation along the direction where the perturbation wave propagates. The study of Burger's flow leads to the conclusion that the domains of temporal instability and spatial instability are parametrically equivalent.

It  has been demonstrated that the equivalence between temporal instability and spatial instability is rather an exception than the rule, by developing the analysis of a Rayleigh-B\'enard system where an anisotropic porous layer saturated by a fluid is heated from below. A permeability anisotropy ratio $\tau$ and the Rayleigh number $R$ are the governing parameters for this flow system. It has been shown that, while the temporal instability can be started up only if the Rayleigh number exceeds its \mbox{$\tau$-dependent} critical value $R_c$, the spatial instability may exist not only when $R > R_c$, but also when $0 < R \le R_c$. This conclusion rules out, in this case, the equivalence between temporal instability and spatial instability proved for the case of Burgers' flow.

The special case where the Fourier modes employed for the spatial stability analysis have both a zero angular frequency and a zero wavenumber has been considered. Such modes are time-independent with an undefined direction of propagation. For the simple example of Burgers' flow, it has been shown that this class of modes may involve a nonzero spatial growth rate even in a parametric domain where the equilibrium solution is both temporally and spatially stable. However, these modes actually display an exponential growth in the positive or in the negative spatial direction. Such a growth can be intended as a form of spatial instability even if not associated to any direction of propagation. The same conclusion has been found also for the Rayleigh-B\'enard system in a saturated anisotropic porous layer. In the study of the saturated porous layer, these stationary Fourier modes with zero wavenumber yet growing in space exist also when the Rayleigh number is negative, {\em i.e.} for conditions of heating from above.

The temporal instability may be studied in a non-modal formulation by testing the time evolution of a perturbation wavepacket expressed via a Fourier integral. This variant of the temporal analysis leads to the definition of a typically supercritical condition called absolute instability \cite{barletta2019routes}. In principle, a similar non-modal conception could be conceived also for the spatial instability. Bringing the spatial instability beyond the purely modal analysis, as well as beyond the merely linear formulation,  could offer interesting opportunities for future research in the fluid mechanics of saturated porous media.

\section*{Acknowledgements}
The author expresses his deepest gratitude to the organising committee of the XXI International Conference on Waves and Stability in Continuous Media, WASCOM 2021 (Catania, Italy: June 6-10, 2022) for the kind invitation and for the financial support.
\\
 The results reported in this paper have been partially presented in WASCOM 2021.
 \\
The author also acknowledges financial support from Italian Ministry of Education, University and Research (MIUR), grant number PRIN 2017F7KZWS.



\providecommand{\BIBde}{de~B}

\end{document}